\documentclass{article}
\usepackage[english]{babel}
\usepackage{authblk}
\usepackage[letterpaper,top=2cm,bottom=2cm,left=3cm,right=3cm,marginparwidth=1.75cm]{geometry}

\usepackage{multirow}
\usepackage[dvipsnames]{xcolor}
\usepackage{amsmath}
\usepackage{mathrsfs}
\usepackage{amssymb}
\usepackage{amsthm}
\usepackage{graphicx}
\usepackage{ulem}
\usepackage[colorlinks=true, allcolors=blue]{hyperref}

\begin{document}
\title{Global Rank Sum Test: An Efficient Rank-Based Nonparametric Test for Large Scale Online Experiment Platform}
\author{Zheng Cai \\ Tencent Technology (Shenzhen) Co., Ltd. \\ zhengcai@tencent.com \and Bo Hu \\ Tencent Technology (Shenzhen) Co., Ltd. \\ lilzedbohu@tencent.com \and Zhihua Zhu \\ Tencent Technology (Shenzhen) Co., Ltd. \\ zhihuazhu@tencent.com}


\maketitle

\begin{abstract}
Online experiments are widely used for improving online services. While doing online experiments, The student t-test is the most widely used hypothesis testing technique. In practice, however, the normality assumption on which the t-test depends on may fail, which resulting in untrustworthy results.
In this paper, we first discuss the question of when the t-test fails, and thus introduce the rank-sum test. Next, in order to solve the difficulties while implementing rank-sum test in large online experiment platforms, we proposed a global-rank-sum test method as an improvement for the traditional one. Finally, we demonstrate that the global-rank-sum test is not only more accurate and has higher statistical power than the t-test, but also more time efficient than the traditional rank-sum test, which eventually makes it possible for large online experiment platforms to use.  
\end{abstract}

\section{Introduction}

Online experiments, also known as A/B tests \cite{kohavi2009controlled}, are widely used for improving online services. Big companies, such as Facebook, Microsoft, Google, Tencent, run thousands of online experiments every day. In an online experiment, users are randomly assigned to one of the two variants: control and treatment, usually abbreviated as $c$ and $t$ respectively. Then after data from each group are collected, a specific hypothesis test will be carried out to decide if the treatment is better than the control. If the answer is yes, the treatment variant will be launched online in place of the control variant.

The \textit{student t-test} is the most widely used hypothesis testing technique. In a t-test, when the sample size is relatively large, we assume that the lift of metrics are normally distributed, and the normal quantiles are used to construct the corresponding confidence intervals. By the \textit{Central Limit Theorem} (CLT) \cite{Bill95}, this normality assumption is satisfied in most cases when the sample sizes are large. 
There are lots of early research discussing what sample size is sufficiently large enough for the t-test, such as \cite{Ratcliffe1968, Barrett1976, Sawilowsky1992, Sullivan1992, lumley2002}. In these research, they recommend around 100 samples are enough for the normality assumption to be valid.
In practice, however, even when the sample size is more than 10,000,000, the normality assumption may still be violated. We explain it in section 2 by demonstrating some empirical examples. We will show that in many real cases, the t-test which relying on this normality assumption has greater than expected type-\uppercase\expandafter{\romannumeral2} error  when this assumption does not hold. 
Thus, t-test would fail to detect some significant improvements resulted in missing some good launches therefore impact users and businesses negatively.



In order to run hypothesis testing when the normality assumption is violated, we switched to using non-parametric tests. The rank-sum test, also known as Mann–Whitney test \cite{MUtest}, is one type of non-parametric tests which is often used to compare samples from two different groups. It uses ranks instead of each sample's value, thus it does not depend on the underlying data distribution. 
Although the rank-sum test avoids the problem when the normality assumption is not valid, 
there's still difficulties when implementing them in large scale experiment platforms. In large scale online experiment platforms, there's a number of experiments running simultaneously. Unlike the fact that while running t-tests, the metric value of a specific user stays the same for different experiments, the rank of the metric from that same user would vary across different experiments. This forces us to rerank sample metrics each time for all the running experiments, resulting in a huge demand for computational resources and making this task impossible in large scale settings.  
In reality, there are thousands of online experiments running on our AB testing platform each day, performing rank-sum test for each one of them separately would be extremely costly. 
As a result, we have mathematically simplified the rank-sum test so that we no longer need to rerank the samples in each experiment, making the application of the rank-sum test feasible in large online experiment platform settings.

The paper is structured as follows. In section 2 we will explain the reason why we need rank-sum test for hypothesis testing, combined with an empirical example. In section 3 we will explain the implementation difficulties of the rank-sum test in large online experiment platforms. In section 4 we introduce our new approach, Global-Rank-Sum test, which makes the application of the rank-sum test feasible and scalable. In section 5 we demonstrate some simulation and empirical results of our new method. Finally, section 6 concludes and makes further discussions.

\section{Motivation}

In this section, we will show the motivation of using rank-sum test instead of t-test by explaining under what circumstances a t-test based on the normality assumption fails, and then illustrate it with some real data. 

First we review the main procedures of t-tests. In the t-test with $\text{H}_0:\mu_t=\mu_c$ against $\text{H}_1:\mu_t\neq \mu_c$, we calculate the sample mean of the control group $\overline{Y_t}$ and treatment group $\overline{Y_c}$ as the estimation of $\mu_t$ and $\mu_c$, respectively. Then we calculate the lift as $\Delta = \overline{Y_t}-\overline{Y_c}$, and the variance of lift as $$\sigma^2 = Var(\Delta)=Var(\overline{Y_t})+Var(\overline{Y_c})=\frac{Var(Y_t)}{N_t}+\frac{Var(Y_c)}{N_c}.$$
Finally we use the \textit{t-}statistics defined as $t = \Delta/\sigma$ for hypothesis testing: if
$|t|>z_{1-\alpha/2}$ then reject the null hypothesis $\text{H}_0$, otherwise accept it, where $1- \alpha$ is the level of confidence and $z_{1-\alpha/2}$ is the $1-\alpha/2$ quantile of the standard normal distribution $N(0,1)$. The reason why we use the normal quantile as the threshold is that we assume $\Delta$ follows the normal distribution, regardless of what distribution users' metrics follow. We call it the \textbf{normality assumption} and it holds from the \textit{Central Limit Theorem} (CLT), which states that suppose $X_1,\cdots,X_n$ are  independent and identically distributed random variables with expectation $\mu$ and finite variance $\sigma^2$, then $\sqrt{n}(\overline{X_n}-\mu)$ converge in distribution to a normal $N(0,\sigma^2)$. By CLT, the sample mean approximately follows the standard normal distribution as the sample size goes larger and larger.

However, in practice, the sample size of an experiment cannot be infinite, so the distribution of the sample mean may not be as close as the normal distribution to utilize some of the properties. Thus the distance between a sample mean distribution and a normal distribution  should be taken into consideration when making the normality assumption.
By Berry–Esseen theorem \cite{Berry1941}, this distance can be bounded by $\frac{C\cdot \rho}{\sigma^3\sqrt{n}}$, where $\rho = E(|X_1|^3)$ is the third absolute moment of $X_1$, and $C$ is a constant. When the original distribution of the metric is long-tailed that takes value in large quantities with more probability, the third absolute moment $\rho$ usually increases faster than $\sigma^3$, so that the distance between the two distributions goes larger.

To give an example, in Tencent Ads, we have got several metrics that are long-tailed. Sometimes the metric value of a few amount of users can be hundreds or even hunderds of thousands of the average value. In that cases, the distribution distance is not negligible even if the sample size is more than ten millions. For example, GMV(Gross Merchandise Value) is an important metric for business, which is defined as the value of users' conversion. Due to various purchasing behaviors of our users, the original distribution of GMV is long-tailed with extremely large values: many users' GMV is zero because they never convert, while some users make high-value purchases, so their GMV are high. Figure \ref{fig:1} shows the distribution of the sample mean for GMV with 10,000,000 samples. We can easily see how it differs from the normal distribution. Therefore, in cases like this, the normality assumption is violated so that the regular t-test based on it are also inaccurate. Empirically, the corresponding $1-\alpha$ confidence interval $(\mu - z_{1-\alpha/2}\cdot \sigma, \mu + z_{1-\alpha/2}\cdot \sigma)$ is wider than expected, thus if running t-test, we would reject fewer null hypotheses and fail to detect some significant improvements.

\begin{figure}
    \centering
    \includegraphics[scale=0.5]{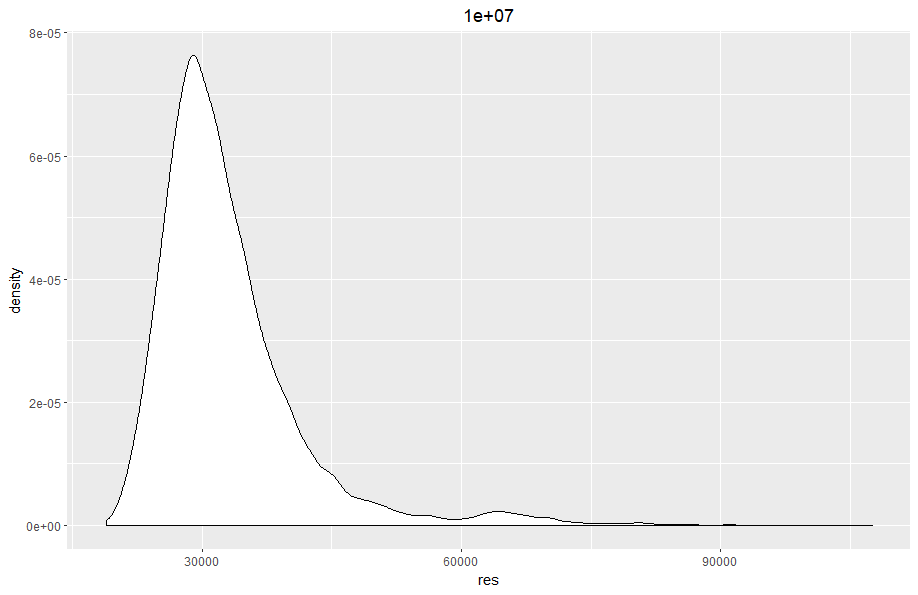}
    \caption{Distribution of sample mean of GMV with 10,000,000 users}
    \label{fig:1}
\end{figure}

In order to run hypothesis tests when the normality assumption is violated, we switched to using non-parametric tests. The rank-sum test is one type of non-parametric tests which is often used to compare samples from two different groups. It uses ranks instead of each sample's value in the rank-sum test, so it does not depend on the distribution of the metrics, which makes it more reliable than the t-test results in some cases.

In detail, let $1,\cdots,N$ be $N$ users, $X_1\cdots,X_N$ be each metric value, and $R_1,\cdots,R_N$ be each value' rank among these $N$ values so that $$R_i = \sum_{j=1}^N I(X_j\leq X_i).$$ 
Without loss of generality, we suppose there are no tied values so that the ranks $R_i$ are different from each other. In the presence of tied values, we may give each tied value a random rank as if we added a random minor perturbation to it to break the tie. 
Let $T_i$ be the group label that $T_i=t$ for user $i$ in the treatment group and $T_i=c$ for user $i$ in the control group. Then the rank-sum statistic \cite{MUtest} is defined as: 
$$rs = \frac{\overline{R_t}-\overline{R_c}}{\sqrt{\frac{N(N^2-1)}{12N_tN_c}}},$$
where $$N_t = \sum_{i=1}^N I(T_i=t),~~~~N_c = \sum_{i=1}^N I(T_i=c),$$
$$\overline{R_t}=\frac{1}{N_t}\sum_{i=1}^NR_iI(T_i=t),~~~~\overline{R_t}=\frac{1}{N_t}\sum_{i=1}^NR_iI(T_i=t).$$
The rank-sum statistic $rs$ can be used in testing the following hypothesis with $H_0$: samples from the two groups follow the same distribution. For a given type-\uppercase\expandafter{\romannumeral1} error $\alpha$, reject the hypothesis if $|rs|>z_{1-\alpha/2}$, accept it otherwise.

\section{Difficulties in Implementing the Rank-Sum Test in Online Experiment Platforms}


In this section, we will introduce the difficulties in implementing the rank-sum test in large-scale online experimental platforms. Let us start with the following example.

\begin{table}
\centering
\begin{tabular}{|c|c|c|c|c|c|c|}
\hline
\multicolumn{1}{|c|}{\multirow{2}{*}{Index}} & \multicolumn{1}{c|}{\multirow{2}{*}{Metric Value}} & \multicolumn{1}{c|}{\multirow{2}{*}{Global Rank}} & \multicolumn{2}{c|}{Experiment 1}                            & \multicolumn{2}{c|}{Experiment 2}                              \\ \cline{4-7} 
\multicolumn{1}{|c|}{}                       & \multicolumn{1}{c|}{}                              & \multicolumn{1}{c|}{} & \multicolumn{1}{c|}{Group Label} & \multicolumn{1}{c|}{Rank} & \multicolumn{1}{c|}{Group Label} & \multicolumn{1}{c|}{Rank} \\ \hline
1                                            & 10                                                 & 4                                                 & t                                & 3                         &                                  &                           \\
2                                            & 9                                                  & 3                                                 & t                                & 2                         &                                  &                           \\
3                                            & 30                                                 & 10                                                & t                                & 6                         &                                  &                           \\
4                                            & 23                                                 & 8                                                 & c                                & 5                         &                                  &                           \\
5                                            & 19                                                 & 7                                                 & c                                & 4                         & t                                & 5                         \\
6                                            & 3                                                  & 1                                                 & c                                & 1                         & c                                & 1                         \\
7                                            & 5                                                  & 2                                                 &                                  &                           & t                                & 2                         \\
8                                            & 27                                                 & 9                                                 &                                  &                           & t                                & 6                         \\
9                                            & 15                                                 & 5                                                 &                                  &                           & c                                & 3                         \\
10                                           & 18                                                 & 6                                                 &                                  &                           & c                                & 4       \\\hline  
\end{tabular}
\caption{Example of traditional rank-sum test procedure}
\label{tbl:1}
\end{table}

As shown in table \ref{tbl:1}, suppose we have 10 samples in total, and each metric value is shown in the same table. Suppose we have two different experiments. Users 1 to 6 are in experiment 1 and users 5 to 10 are in experiment 2. In order to run t-test for each experiment, we only need to calculate the t-statistic by the corresponding metric value. But when we run the rank-sum tests, we need to sort the samples in each experiment before calculating the rank sum statistics. Since different experiments contain different sets of users, we need to sort the metrics twice. Let's say we have $E$ experiments, each containing $N$ samples, then we have to run some $N$-element sorting $M$ times. Since the time complexity of an $N$-element sorting is $o(N\log (N))$, the total time complexity is $o(E\cdot N\log (N))$.  In Tencent Ads, the sample size $N \approx 10^8$. Sorting such large crowd is often time-consuming. In the online experiment platform of Tencent Ads, such sorting usually takes about an hour. However, we have thousands of experiments running at the same time, and it would take thousands of hours to run the sorting for each experiment separately, which is unacceptable. Thus we have to find other ways to sort the samples for each experiment.


An intuitive idea comes from the fact that although  the metric rank of a same user differs in different experiments, the relative rank among all samples is consistent. Therefore, we may simplify the sorting procedures by the following two steps. First, all samples are sorted to obtain a global rank. Second, for each experiment, we filter the samples in the experiment from the globally sorted samples in order. In that way, we could do an easier sort to obtain the rank of each sample for that specific experiment. However, this method is still time-consuming, because the second step runs $E$ times, and not convenient for distributed computing.


Before exploring this area, we tried to find any previous research related to this topic, but there were very few studies on non-parametric tests nor on their applications in large-scale experiment platforms. 
\cite{abhishek2017nonparametric} provides a detailed explanation of one type of non-parametric methods in online experiments. This testing method uses bootstrap \cite{bradley1993introduction} without relying on the distribution of data and has been successfully used on some complex metrics which provides some inspirations to our work. However, when faced with a large number of online experiments that need to be tested simultaneously, this method may also encounter some engineering challenges that prevent it from providing immediate and efficient result to experiment. 



\section{The Global-Rank-Sum Test}
\newtheorem{thm}{Theorem}
In this section we introduce our new method for rank-sum test. In our method, we only use the global rank of each sample for testing each experiment.

Let $\mathbb{I}= \{1,2,\cdots ,N\}$ be the total user population of size $N$. Let 
$\mathbb{J}\subseteq \mathbb{I}$ be the users involved in an experiment with sample size $M$. $\mathbb{J}$ is randomly split into the treatment group $\mathbb{J}_t$ and control group $\mathbb{J}_c$, with sample size $N_t$ and $N_c$, respectively, so that $N_t + N_c = M$. 
Let $T_i$ be $i$'s group: $T_i = t$ stands for user $i$ in the treatment group $\mathbb{J}_t$, $T_i = c$ stands for $i$ in the control group $\mathbb{J}_c$, and $T_i = o$ stands for other users not in the experiment. In summary,
\[
T_i = 
\begin{cases}
t, & i \in \mathbb{J}_t \\
c, & i \in \mathbb{J}_c\\
o, & i \in \mathbb{I}\backslash \mathbb{J}\\
\end{cases}.
\]
Let $X_i$ be user $i$'s metric value. Let $R_i$ be user $i$'s global rank among the total population $\mathbb{I}$, i.e., $$R_i = \sum_{j=1}^N I(X_j\leq X_i).$$

\begin{thm}
Given that samples are randomly split into the treatment group and control group, i.e., for some constant $0<p<1$, $P(T_i = t\mid T_i = t \text{ or } c) = p$ holds for all $i$, then under the null hypothesis $H_0$: samples in $\mathbb{J}_t$ and $\mathbb{J}_c$ follow the same distribution, the statistic
\[
grs =  \frac{\overline{R_t}-\overline{R_c}}{\sigma\cdot \sqrt{\frac{1}{N_t} + \frac{1}{N_c}}}\stackrel{d}{\rightarrow} N(0,1)
\]
as $N_t, N_c\rightarrow \infty$, where 
\[\overline{R_t} = \frac{1}{N_t} \sum_{i\in \mathbb{J}_t} R_i,~~~~~\overline{R_c} = \frac{1}{N_c} \sum_{i\in \mathbb{J}_c} R_i,~~~~\overline{R}=\frac{1}{M} \sum_{i\in \mathbb{J}} R_i,~~~~\sigma^2 = \frac{1}{M-1} \cdot \sum_{i \in \mathbb{J}}(R_i - \overline{R})^2.\]
\end{thm} 
\begin{proof}
Denote $\mathbb{J}=\{j_1,\cdots,j_M\}$. Let $\mathbb{R} = \{R_i:~i\in \mathbb{J}\}$ be the rank of all users in the experiment.  Let $\mathbb{P}=\{(r_1,\cdots,r_M):\ (r_1,\cdots,r_M)\text{ is a permutation of }(R_{j_1},\cdots,R_{j_M})\}$ be the set of all the $M!$ permutations of $\mathbb{R}$. 
Then under the null hypothesis, the ranks of all $R_{i}\in \mathbb{R}$'s takes values in $\mathbb{P}$ with the same probability, i.e.
\[P((R_{j_1},\cdots,R_{j_M})=(r_1,\cdots,r_M))=\frac{1}{M!}\]
holds for all $(r_1,\cdots,r_M)\in\mathbb{P}$. Thus, $P(R_i = r_j) = 1/M$ for all $i$ and $j$, and $E(R_i) = \frac{1}{M}\sum_{j=1}^M r_j$ holds for all $i$. Therefore,
\[E(\overline{R_t}-\overline{R_c}) = \frac{1}{M}\sum_{j=1}^M r_j - \frac{1}{M}\sum_{j=1}^M r_j = 0.\]
It suffices to prove that
\[
Var(\overline{R_t}-\overline{R_c}) = 
\left(\frac{1}{N_t} + \frac{1}{N_c}\right)\cdot \sigma^2,
\]
then by the CLT, we have
\[
grs = \frac{E(\overline{R_t}-\overline{R_c})}{\sqrt{E(\overline{R_t}-\overline{R_c})}} \stackrel{d}{\rightarrow} N(0,1),
\]
which completes the proof. The details of calculation of $Var(\overline{R_T}-\overline{R_C})$ is as follows.

For any $i$, we have
\[
Var(R_i) = E(R_i^2)-(E(R_i))^2=\frac{1}{M}\sum_{i=1}^Mr_i^2 - \left(\frac{1}{M}\sum_{i=1}^Mr_i\right)^2=\frac{1}{M}\sum_{i=1}^M(r_i-\overline{r})^2=\frac{M-1}{M}\sigma^2.
\]
For any $i_1\neq i_2$, we have
\[
E(R_{i_1}R_{i_2}) = \frac{1}{M(M-1)}\sum_{1\leq j_1\neq j_2\leq N}r_{j_1}r_{j_2} = \frac{1}{M(M-1)}\left(\left(\sum_{i=1}^Mr_i\right)^2-\sum_{i=1}^Mr_i^2\right),
\]
so that
\[
Cov(R_{i_1},R_{i_2})=E(R_{i_1}R_{i_2})-E(R_{i_1})E(R_{i_2})=\frac{1}{M(M-1)}\left(\left(\sum_{i=1}^Mr_i\right)^2-\sum_{i=1}^Mr_i^2\right)-\left(\frac{1}{M}\sum_{j=1}^M r_j\right)^2=-\frac{1}{M}\sigma^2.
\]
Thus,
\[
Var(\overline{R_t})=\frac{1}{N_t^2}Var\left(\sum_{i\in \mathbb{J}_t}r_i\right)=\frac{1}{N_t^2}\left(N_t\cdot Var(R_i) + N_t(N_t-1)\cdot Cov(R_{i_1},R_{i_2})\right)=\frac{N_c}{N_tM}\sigma^2,
\]
\[
Var(\overline{R_c})=\frac{1}{N_c^2}Var\left(\sum_{i\in \mathbb{J}_c}r_i\right)=\frac{1}{N_c^2}\left(N_c\cdot Var(R_i) + N_c(N_c-1)\cdot Cov(R_{i_1},R_{i_2}\right)=\frac{N_t}{N_cM}\sigma^2.
\]
Finally, we have
\[
Var(\overline{R_t}-\overline{R_c})=Var(\overline{R_t})+Var(\overline{R_c})-2Cov(\overline{R_t},\overline{R_c})=\frac{N_c}{N_tM}\sigma^2+\frac{N_t}{N_cM}\sigma^2+\frac{2}{M}\sigma^2=\left(\frac{1}{N_t} + \frac{1}{N_c}\right)\cdot \sigma^2
\]
\end{proof}
\section{Simulation Study}
\subsection{contrast type-I error of t-test and rank-sum tests}
In this subsection, we carry out simulations based on the log-normal distributed data so that each user's metric $Y_i$ is independently drawn from the distribution
\[log(Y_i) \sim N(\mu, \sigma^2)\]
with some fixed parameters $(\mu, \sigma)$. We set different choices of parameters as shown in Table \ref{tbl:2}. First, with each $(\mu, \sigma)$, we randomly generate a total population of $N$ users with $N$ fixed as $1,000,000$. Next, we simulate 5,000 experiments, each with $100,000$ users in the treatment group and $100,000$ users in the control group randomly selected from the total population. For each experiment, we run a t-test, a traditional rank-sum test, and a global-rank-sum test as we introduced in this paper. Finally, we summarised the proportion significant tests with type-\uppercase\expandafter{\romannumeral1} error $\alpha=0.01,0.05,0.1$. The results are summarised in Table \ref{tbl:2}.

\begin{table}
\centering
\begin{tabular}{|c|c|c|c|c|c|c|c|c|c|}
\hline
\multirow{2}{*}{$(\mu,\sigma)$} & \multicolumn{3}{c|}{t-test} & \multicolumn{3}{c|}{rank-sum test} & \multicolumn{3}{c|}{global-rank-sum test} \\ \cline{2-10} 
                               & $\alpha$=0.01  & $\alpha$=0.05  & $\alpha$=0.10  & $\alpha$=0.01  & $\alpha$=0.05  & $\alpha$=0.10  & $\alpha$=0.01  & $\alpha$=0.05  & $\alpha$=0.10   \\ \hline
(-3,3)                          & 0.42\%                         & 4.30\%                   & 9.52\%                  & 1.26\%                   & 5.44\%                   & 10.26\%                 & 1.26\%                   & 5.42\%                   & 10.24\%                 \\
(-3,4)                          & 0.06\%                         & 2.40\%                   & 7.84\%                  & 1.14\%                   & 5.32\%                   & 10.30\%                 & 1.16\%                   & 5.32\%                   & 10.30\%                 \\
(-3,5)                          & 0.06\%                         & 2.10\%                   & 6.72\%                  & 0.96\%                   & 5.18\%                   & 10.62\%                 & 0.96\%                   & 5.16\%                   & 10.64\%                 \\
(-3,6)                          & 0.02\%                         & 0.98\%                   & 4.82\%                  & 0.92\%                   & 5.38\%                   & 9.82\%                  & 0.92\%                   & 5.36\%                   & 9.84\%                  \\
(-3,7)                          & 0.04\%                         & 0.94\%                   & 4.58\%                  & 1.34\%                   & 5.48\%                   & 10.28\%                 & 1.32\%                   & 5.48\%                   & 10.30\%                 \\\hline
(-4,3)                          & 0.56\%                         & 4.38\%                   & 9.70\%                  & 1.10\%                   & 5.38\%                   & 10.54\%                 & 1.10\%                   & 5.38\%                   & 10.50\%                 \\
(-4,4)                          & 0.22\%                         & 2.96\%                   & 9.04\%                  & 1.40\%                   & 5.86\%                   & 10.46\%                 & 1.40\%                   & 5.90\%                   & 10.44\%                 \\
(-4,5)                          & 0.04\%                         & 1.16\%                   & 4.54\%                  & 0.96\%                   & 4.76\%                   & 9.60\%                  & 0.96\%                   & 4.74\%                   & 9.56\%                  \\
(-4,6)                          & 0.06\%                         & 2.32\%                   & 7.42\%                  & 1.02\%                   & 4.82\%                   & 9.68\%                  & 1.02\%                   & 4.82\%                   & 9.76\%                  \\
(-4,7)                          & 0.00\%                         & 1.50\%                   & 4.82\%                  & 1.04\%                   & 5.14\%                   & 9.90\%                  & 1.04\%                   & 5.12\%                   & 9.88\%                  \\\hline
(-5,3)                          & 0.44\%                         & 3.74\%                   & 9.26\%                  & 1.28\%                   & 5.52\%                   & 10.14\%                 & 1.28\%                   & 5.56\%                   & 10.16\%                 \\
(-5,4)                          & 0.34\%                         & 2.72\%                   & 6.96\%                  & 0.74\%                   & 4.64\%                   & 9.48\%                  & 0.76\%                   & 4.64\%                   & 9.52\%                  \\
(-5,5)                          & 0.06\%                         & 1.06\%                   & 4.72\%                  & 1.12\%                   & 5.48\%                   & 10.62\%                 & 1.10\%                   & 5.50\%                   & 10.62\%                 \\
(-5,6)                          & 0.12\%                         & 1.46\%                   & 4.28\%                  & 1.04\%                   & 4.94\%                   & 9.38\%                  & 1.04\%                   & 4.94\%                   & 9.42\%                  \\
(-5,7)                          & 0.02\%                         & 0.52\%                   & 3.16\%                  & 1.02\%                   & 5.16\%                   & 10.00\%                 & 1.02\%                   & 5.16\%                   & 9.98\% \\\hline
\end{tabular}
\caption{Type-\uppercase\expandafter{\romannumeral1} errors of t-test, rank-sum test and global-rank-sum test}
\label{tbl:2}
\end{table}

As the result shown in Table \ref{tbl:2}, when the distribution is extremely long-tailed, such as in the situation when $-\mu$ and $\sigma$ are relatively large, the t-test has less type-\uppercase\expandafter{\romannumeral1} error than expected. 
Figure \ref{fig:2} shows the distribution of the sample mean $\overline{Y_i}$ of $100,000$ samples, in which we can see that the distribution is not so close to the normal distribution.
\begin{figure}
    \centering
\begin{tabular}{cc}
   \includegraphics[scale=0.27]{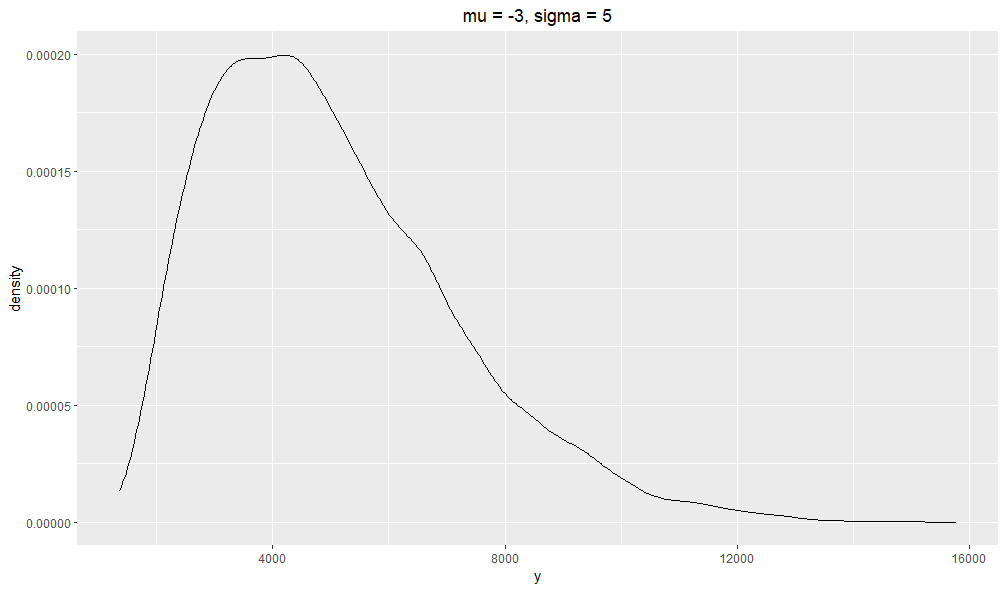}  & \includegraphics[scale=0.27]{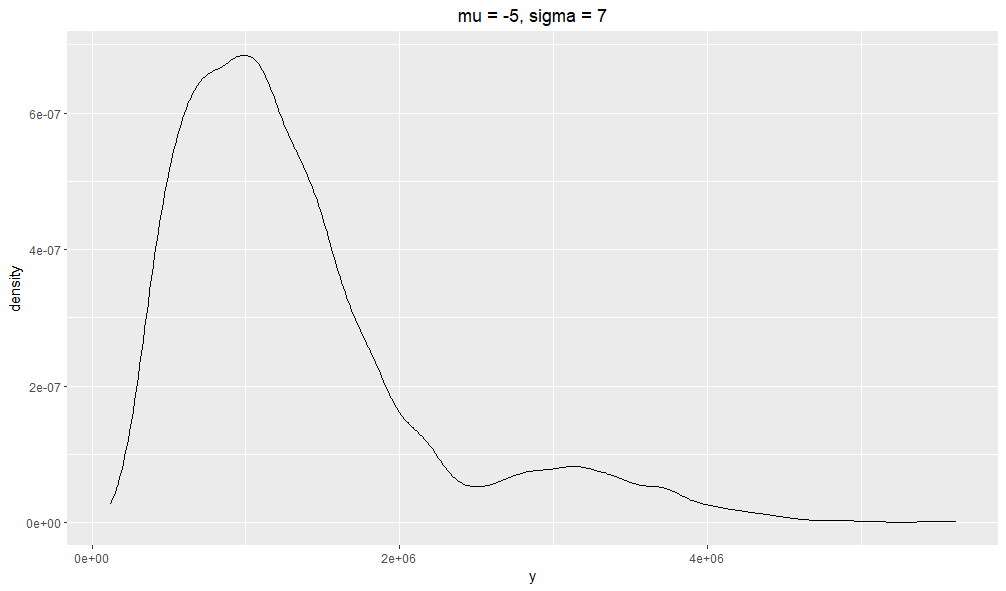} \\
   (a) $\mu=-3, \sigma=5$  & (b) $\mu=-5, \sigma=7$
\end{tabular}
    \caption{Distribution of sample mean when the original distribution is extremely long-tailed.}
    \label{fig:2}
\end{figure}
At the same time, even in the most extreme situations, we can see that the type-\uppercase\expandafter{\romannumeral1} error of the rank-sum test and global-rank-sum test are both very close to the prior $\alpha$, as we expected. Thus, in the extremely long-tailed situations, the rank-sum tests are more accurate than the t-test.

\subsection{contrast power of t-test and rank-sum tests}
We have shown that the type-\uppercase\expandafter{\romannumeral1} error of t-test is lower than expected in the extremely long-tailed situation, which indicates that the t-test rejects less hypothesis, resulting in higher type-\uppercase\expandafter{\romannumeral2} error. In this subsection, we will contrast the power of t-test and rank-sum tests.

Here we fix the parameters $(\mu,\sigma)=(-5,7)$, and set the user metrics $Y_i$ in the treatment group lifted by some fixed ratio, namely,
\[Y_i^{\text{obs}} = (1+\gamma \cdot I(T_i=t))\cdot Y_i,\] where $\gamma$ is the lift ratio. Other settings are identical to those in the previous simulation.

As shown in Table \ref{tbl:3}, we see the power of the rank-sum test and the global-rank-sum test are almost the same, and both of them are much greater than the power of the t-test. Combining this result with the contrast of type-\uppercase\expandafter{\romannumeral1} error, we demonstrate that the rank-sum test is more accurate and powerful than the t-test in the situation of extremely long-tailed distribution.

\begin{table}
\centering
\begin{tabular}{|c|c|c|c|c|c|c|c|c|c|}
\hline
\multirow{2}{*}{lift ratio $\gamma$} & \multicolumn{3}{c|}{t-test} & \multicolumn{3}{c|}{rank-sum test} & \multicolumn{3}{c|}{global-rank-sum test} \\ \cline{2-10} 
                               & $\alpha$=0.01  & $\alpha$=0.05  & $\alpha$=0.10  & $\alpha$=0.01  & $\alpha$=0.05  & $\alpha$=0.10  & $\alpha$=0.01  & $\alpha$=0.05  & $\alpha$=0.10   \\ \hline
1\%                         & 0.02\%  & 0.92\%  & 4.16\% & 1.20\%    & 6.34\%    & 11.72\%   & 1.20\%      & 6.34\%       & 11.72\%     \\
2\%                         & 0.00\%  & 0.62\%  & 2.98\% & 2.82\%    & 9.38\%    & 15.88\%   & 2.84\%      & 9.38\%       & 15.88\%     \\
3\%                         & 0.04\%  & 2.12\%  & 8.70\% & 4.94\%    & 15.28\%   & 24.74\%   & 4.96\%      & 15.28\%      & 24.72\%     \\
4\%                         & 0.02\%  & 0.64\%  & 4.44\% & 9.00\%    & 23.70\%   & 34.54\%   & 8.96\%      & 23.76\%      & 34.48\%     \\
5\%                         & 0.06\%  & 0.92\%  & 3.52\% & 14.66\%   & 33.56\%   & 45.52\%   & 14.68\%     & 33.52\%      & 45.54\%     \\
10\%                        & 0.02\%  & 0.84\%  & 4.16\% & 64.46\%   & 84.20\%   & 90.12\%   & 64.44\%     & 84.20\%      & 90.12\%     \\
20\%                        & 0.04\%  & 1.08\%  & 3.84\% & 99.88\%   & 100.00\%  & 100.00\%  & 99.88\%     & 100.00\%     & 100.00\%   \\\hline

\end{tabular}
\caption{Power of t-test, rank-sum test and global-rank-sum test}
\label{tbl:3}
\end{table}
\subsection{comparative speed analysis}
In the last subsection we contrast the time efficiency of the traditional rank-sum test and the global-rank-sum test. Although this comparison is intuitively obvious, since the global-rank-sum test only need to sort the data once instead of thousands of times, here we still confirm it with the simulation data. 

Table \ref{tbl:4} summarises the time cost of the two types of rank-sum tests. When there are only one experiment, the traditional rank-sum test costs less time than the global-rank-sum test, because it only needs to sort the samples in the experiment instead of the total population. As the number of experiments increases, the advantage in time efficiency of the global-rank-sum test over the traditional rank-sum test also increases. 

In summary, through the simulation of the type-\uppercase\expandafter{\romannumeral1} error, statistical power, and time efficiency, we can see that in the case of extremely long-tailed distribution, the rank-sum tests are more accurate and has higher statistical power than the t-test, and the global-rank-sum test is further more efficient than the traditional rank-sum test.

\begin{table}
    \centering
   \begin{tabular}{|c|c|c|c|}
   \hline
Number of & seconds of             & seconds of                    & \multirow{2}{*}{time cost diff ratio} \\
Experiments          & rank-sum tests & global-rank-sum tests &                      \\\hline
1         & 0.092                    & 0.386                           & 319.6\%              \\
10        & 0.784                    & 0.644                           & -17.9\%              \\
50        & 3.008                    & 1.415                           & -53.0\%              \\
100       & 6.183                    & 2.211                           & -64.2\%              \\
200       & 11.541                   & 3.748                           & -67.5\%              \\
500       & 30.057                   & 9.025                           & -70.0\%             \\\hline
\end{tabular}
    \caption{Time cost of traditional rank-sum test and global-rank-sum test}
    \label{tbl:4}
\end{table}

\section{Conclusion and Further Discussion}
In this paper, we first discuss the question of when the t-test fails, and thus introduce the rank-sum test. Next, in order to solve the implementing difficulties of the rank-sum test in the online experiment platform, we proposed a global-rank-sum test method as an improvement of the traditional rank-sum test. Finally, we demonstrate that the global-rank-sum test is not only more accurate and has higher statistical power than the traditional t-test, but also more efficient than the traditional rank-sum test, which makes it feasible to implement the rank-sum test in the online experiment platform.

Our designed global-rank-sum test does not require additional sorting given that the finest units are sorted (typically individual level). This is suitable and convenient for ad-hoc post experiment analysis involving data slicing or filtering, which is very common when carrying out AB tests and the test designers would like to see granular level results for example differences between two groups' people with specific gender, location or certain characteristics. In this situation, one could carry out a non-parametric test by picking the ranks of the selected people and repeat the procedure.
Also, non-parametric tests such as the rank-sum test could be a double edged sword. Unlike regular hypothesis testing procedures, which gives information about the underlying distribution (the mean and the variance), the non-parametric tests' statistics often are not as informative as its counterparts in a regular hypothesis test. A mediation metric, such as median or quantile, could be calculating some intermediate measures while performing the rank-sum tests. In addition, our new method only deals with the situation without tied values. In the situation with tied values, we put those users in random different rank to break the tie. A more refined approach would be to make a tie correction for the global-rank-sum statistic, which is worthy of further exploration.

\bibliographystyle{acm}
\bibliography{reference}
\end{document}